\documentclass[11pt]{article}

\usepackage{amsmath}
\usepackage{amssymb}
\usepackage{fullpage}
\usepackage{color}
\usepackage{graphicx}
\usepackage{epsfig}
\usepackage{amsthm}
\usepackage{latexsym}
\usepackage{amssymb}
\usepackage{amsmath}
\usepackage{verbatim}
\usepackage{graphicx}
\usepackage{epsfig}
\usepackage{subfigure}
\usepackage{cite}
\usepackage{url}
\usepackage[table]{xcolor}
\usepackage{tabularx}
\usepackage{algorithmic}
\usepackage{algorithm}
\usepackage{bbm}

\newtheorem{theorem}{Theorem}

\newtheorem{proposition}{Proposition}

\newtheorem{lemma}{Lemma}

\newtheorem{definition}{Definition}

\newcommand{\A}{\mathcal{A}}

\newcommand{\EF}{\mathrm{EF}}
\newcommand{\EFone}{\mathrm{EF1}}
\newcommand{\EFX}{\mathrm{EFX}}
\newcommand{\EFL}{\mathrm{EFL}}
\newcommand{\bEFL}{\mathrm{\mathbf{EFL}}}

\newcommand{\MMS}{\mathrm{MMS}}

\newcommand{\GMMS}{\mathrm{GMMS}}
\newcommand{\bGMMS}{\mathrm{\mathbf{GMMS}}}
\newcommand{\PMMS}{\mathrm{PMMS}}

\newcommand{\B}{\mathcal{B}}

\DeclareMathOperator*{\argmax}{arg\,max}
\DeclareMathOperator*{\argmin}{arg\,min}

\begin{document}
\title{{\bfseries Groupwise Maximin Fair Allocation of Indivisible Goods}}
\author{
{Siddharth Barman}\thanks{Indian Institute of Science. {\tt barman@iisc.ac.in}}
\quad {Arpita Biswas}\thanks{Indian Institute of Science. {\tt arpitab@iisc.ac.in}}
\quad {Sanath Kumar Krishnamurthy}\thanks{Chennai Mathematical Institute. {\tt sanathkumar9@cmi.ac.in}}
\quad {Y. Narahari}\thanks{Indian Institute of Science. {\tt narahari@iisc.ac.in}}
}
\date{}
\maketitle

\begin{abstract}
We study the problem of allocating \emph{indivisible} goods among $n$ agents in a \emph{fair} manner. For this problem, maximin share ($\MMS$) is a well-studied solution concept which provides a fairness threshold. Specifically, maximin share is defined as the minimum utility that an agent can guarantee for herself when asked to partition the set of goods into $n$ bundles such that the remaining ($n-1$) agents pick their bundles adversarially. An allocation is deemed to be fair if every agent gets a bundle whose valuation is at least her maximin share.  

Even though maximin shares provide a natural benchmark for fairness, it has its own drawbacks and, in particular, it is not sufficient to rule out unsatisfactory allocations. Motivated by these considerations, in this work we define a stronger notion of fairness, called \textit{groupwise maximin share guarantee} ($\GMMS$). In $\GMMS$, we require that the maximin share guarantee is achieved not just with respect to the grand bundle, but also among all the subgroups of agents. Hence, this solution concept strengthens $\MMS$ and provides an \emph{ex-post} fairness guarantee. We show that in specific settings, $\GMMS$ allocations always exist. We also establish the existence of approximate $\GMMS$ allocations under additive valuations, and develop a polynomial-time algorithm to find such allocations. Moreover, we establish a scale of fairness wherein we show that $\GMMS$ implies approximate envy freeness. 

Finally, we empirically demonstrate the existence of $\GMMS$ allocations in a large set of randomly generated instances. For the same set of instances, we additionally show that our algorithm achieves an approximation factor better than the established, worst-case bound. 
\end{abstract}

\section{Introduction}
In recent years, the topic of fair division of \textit{indivisible goods} has received significant attention in the computer science, mathematics, and economics communities, see, for instance, Chapter 12 in \cite{brandt2016handbook}. A central motivation behind this thread of research is the fact that classical notions of fairness---such as \textit{envy freeness} ($\EF$)\footnote{An allocation is called envy free if every agent values her bundle at least as much she values any other agent's bundle \cite{foley1967resource,varian1974equity,stromquist1980cut}.} and \textit{proportionality}\footnote{An allocation, among $n$ agents, is called proportionally fair if every agent's value for her share is at least $1/n$ times the total value of all the goods~\cite{steinhaus1948problem}.}---which were developed for divisible goods (that can be fractionally allocated), do not translate well to the indivisible case. For instance, if there is a single indivisible good and two agents, then no allocation can guarantee $\EF$ or {proportionality}. But, given that a number of real-world settings (such as budgeted course allocations \cite{budish2011combinatorial}, division of inheritance, and partitioning resources in a cloud computing environment) entail allocation of discrete/indivisible goods, it is essential to define and study solution concepts which are applicable for a fair division of indivisible goods. 

Classically, the applicability of solution concepts is studied via existence results. Understanding if and when a solution concept is guaranteed to exist is of  fundamental importance in microeconomics and other related fields. Such existence results have been notably complemented by research---in algorithmic game theory and artificial intelligence---that has focused on computational issues surrounding the underlying solution concepts. Broadly, our results contribute to these key themes by establishing both existential and algorithmic results for a new notion of fairness. 

In the context of fair division, the focus on developing efficient algorithms is motivated, in part, by websites such as \textit{Spliddit}\footnote{www.spliddit.org} \cite{goldman2015spliddit} and \textit{Adjusted Winner}\footnote{http://www.nyu.edu/projects/adjustedwinner/} \cite{brams2005efficient}, which offer fair solutions for dividing goods. Spliddit has attracted more than fifty thousand users and, among other services, it computes allocations which are fair with respect to the standard notions of fairness. One of the solution concepts considered by Spliddit is the \emph{maximin share guarantee} ($\MMS$). 

The $\MMS$ solution concept was defined in the notable work of \cite{budish2011combinatorial}, and it deems an allocation to be fair if each agent gets a bundle of value greater than or equal to an agent-specific threshold, called the \textit{maximin share} of the agent. Specifically, the maximin share of an agent corresponds to the maximum value that the agent can attain for herself if she is hypothetically asked to partition the set of goods into $n$ bundles and, then, the remaining ($n - 1$) agents pick their bundles adversarially. Hence, a risk-averse agent $i$ would partition the goods (into $n$  bundles) such that value of the least desirable bundle (according to her) in the partition is maximized. The value of the least desirable  bundle in such a partition is called the \textit{maximin share} of agent $i$. This definition can be interpreted as a natural generalization of the classical \emph{cut-and-choose} protocol.

Although maximin shares provide a natural benchmark to define fairness, this solution concept has its own drawbacks. In particular, $\MMS$ is not sufficient to rule out unsatisfactory allocations; see Section \ref{sec:example} for an example. Moreover, different $\MMS$ allocations can be drastically different in terms of, say, the social welfare of the agents.

Motivated by these considerations, we define a strictly stronger notion of fairness, called \textit{groupwise maximin share guarantee} ($\GMMS$). Intuitively, $\GMMS$ provides an \emph{ex-post} fairness guarantee: it ensures that, even after the allocation has been made, the maximin share guarantee is achieved not just with respect to the grand bundle of goods, but also among all the subgroups of agents $J\subseteq [n]$. Specifically, we say that an allocation is $\GMMS$ if, for all groups $J \subseteq [n]$ and agents $i \in J$, the value of $i$'s bundle in the allocation is no less than the maximin share of $i$ restricted to $J$. That is, if the agent $i$ were to repeat the thought experiment (of dividing all the goods allocated to the agents in group $J$, so that the other $|J|-1$ agents pick their bundles adversarially) to calculate her maximin share restricted to $J$, then the value of her bundle is at least this threshold. This definition directly ensures that groupwise maximin share guarantee is a stronger solution concept: $\GMMS$ implies $\MMS$. In Section \ref{sec:example}, we show that $\GMMS$ can, in fact, be arbitrarily better than an allocation that just satisfies $\MMS$. 

$\GMMS$ also strictly generalizes \textit{pairwise maximin share guarantee} ($\PMMS$), a notion defined by \cite{caragiannis2016unreasonable}. In $\PMMS$, the maximin share guarantee is required only for pairs of agents, but not necessarily for the grand bundle. Section \ref{example:GMMS_PMMS_MMS}, provides an example which establishes that $\GMMS$ is a strict generalization of $\PMMS$ and $\MMS$.

The relevance of $\GMMS$ is also substantiated by the fact that it implies other complementary notions of fairness, which do not follow from $\MMS$ alone. In the context of indivisible goods, relaxations of envy freeness, such as $\EFone$\footnote{An allocation is said to be envy-free up to one good ($\EFone$) if no agent envies any other after removing at most one good from the other agent's bundle; see Definition \ref{def:EF1}.} \cite{budish2011combinatorial} and $\EFX$\footnote{An allocation is said to be envy-free up to the least valued good ($\EFX$) if no agent envies any other agent after removing any positively valued good from the other agent's bundle; see Definition~\ref{def:EFX}.} \cite{caragiannis2016unreasonable} have also been studied. In Section \ref{sec:fairness-scale} we show that (unlike $\MMS$) $\GMMS$ fits into this scale of fairness and, in particular, a $\GMMS$ allocation is guaranteed to be $\EFX$ (and hence $\EFone$). These implications essentially follow from the observation that, by definition, a $\GMMS$ allocation is $\PMMS$ as well.  \cite{caragiannis2016unreasonable} have shown that $\PMMS$ implies $\EFX$, and hence we obtain the desired implications.

Throughout the paper, we focus on additive valuations, and a high-level contribution of our work is to show that under additive valuations, a number of useful (existence, algorithmic, and approximation) results which have been established for $\MMS$ continue to hold for $\GMMS$ as well. 

\subsection{Our Contributions}
In addition to proving a scale of fairness for $\GMMS$, we establish the following results:
\begin{enumerate}
\item Approximate groupwise maximin share allocations always exist under additive valuations. Prior work has shown that there are instances wherein no allocation is $\MMS$~\cite{procaccia2014fair,kurokawa2016can}. These non-existence results have motivated a detailed study of approximate $\MMS$ allocations, i.e., allocations in which each agent gets a bundle of value (multiplicatively) close to her maximin share. Along these lines, we consider approximate $\GMMS$ (see Definition \ref{def:approx-GMMS}), and show that, under additive valuations, a $1/2$-$\GMMS$ allocation always exists. In addition, such an allocation can be found in polynomial time. 
\item $\GMMS$ allocations are guaranteed to exist when the valuations of the agents are either binary or identical. 
\item Analogous to the experimental results for $\MMS$ \cite{bouveret2014characterizing}, $\GMMS$ allocations exist empirically. These simulation results indicate that we do not fall short on such generic existence results by strengthening the maximin solution concept.
\end{enumerate}

\subsection{Related Work}
As mentioned earlier, \cite{budish2011combinatorial} introduced the notion of maximin share guarantee ($\MMS$), and it has been extensively studied since then. In particular, \cite{bouveret2014characterizing} showed that if the agents' valuations are additive, then an envy free (or proportional) allocation will be $\MMS$ as well. They also established that $\MMS$ exists under binary, additive valuations. Their experiments, using different distributions over the valuations, did not yield a single example wherein an $\MMS$ allocation did not exist.  \cite{kurokawa2016can} showed that $\MMS$ allocations exist with high probability when valuations are drawn randomly.

However, \cite{procaccia2014fair} provided intricate counterexamples to refute the universal existence of $\MMS$ allocations, even under additive valuations. This motivated the study of approximate maximin share allocations, $\alpha$-$\MMS$, where each agent obtains a bundle of value at least $\alpha \in (0,1)$ times her maximin share. \cite{procaccia2014fair} established the existence of $2/3$-$\MMS$, and developed a polynomial-time algorithm to obtain such an allocation when the number of agents is a fixed constant. Later, \cite{amanatidis2015approximation} showed that a $2/3$-$\MMS$ can be computed in polynomial (in the number of players) time. 

Approximate maximin share allocations have also been studied for general valuations. \cite{barman2017approximation} have developed an efficient algorithm which finds a $1/10$-$\MMS$ allocation under submodular valuations. More recently, \cite{ghodsi2017fair} have improved the approximation guarantee for additive valuations to $3/4$. They have also developed constant-factor approximation guarantees for submodular and \emph{XOS} valuations, along with a logarithmic approximation for subadditive valuations.

\cite{aziz2016approximation} studied the fair division of indivisible \textit{chores} (negatively valued goods) and have developed an efficient algorithm which finds a $2$-$\MMS$ allocation.

\cite{caragiannis2016unreasonable} defined another important fairness notion called \textit{pairwise maximin share guarantee} ($\PMMS$). As mentioned previously, under $\PMMS$, the maximin share guarantee is required only for pairs of agents, and not even for the grand bundle. They also established that $\PMMS$ and $\MMS$ are incomparable: neither one of these solution concepts implies the other.

\section{Preliminaries and Notation}
We consider the problem of finding a \emph{fair} allocation of a set of \textit{indivisible} goods $[m]= \{1,\ldots,m\}$, among a set of agents $[n]= \{1,\ldots,n\}$. For a subset of goods $S~\subseteq~[m]$ and integer $t$, let $\Pi_t(S)$ denote the set of all $t$ partitions of $S$. An \textit{allocation} is defined as an $n$-partition $(A_1,A_2,\dots,A_n) \in \Pi_n([m])$, where $A_i$ is the set of goods allocated to agent $i$.

The preference of the agents over the goods is specified via valuations. Specifically, we denote the valuation of an agent $i \in [n]$ for a subset of goods $S\subseteq[m]$ by $v_{i}(S)$. In this work, we assume the valuations to be \textit{non-negative} and \textit{additive}, i.e., $v_i\left(\{g\}\right)\geq 0$ for all $g\in [m]$ and $v_i(S) = \sum_{g\in S}\ v_i\left(\{g\}\right)$. For ease of presentation, we will use $v_i(g)$ for agent $i$'s valuation of good $g$, i.e., for $v_i(\{ g \})$. 

As mentioned previously, the fairness notions considered in this work are defined using thresholds called maximin shares. Formally, given an agent $i$, parameter $k \in \mathbb{Z}_+$, and subset of goods $S \subseteq [m]$, $k$-maximin share of $i$ restricted to $S$ is defined as $\mu_i^{k}(S) := \max_{(P_1,\ldots, P_k)\in \Pi_{k}(S)}\  \min_{j\in [k]}\ v_i(P_j)$.
Throughout, $\MMS_i $ will be used to denote the maximin share of an agent $i$ with respect to the grand bundle, $\MMS_i := \mu_i^n([m])$.
We will now formally define \textit{maximin share allocation} ($\MMS$ allocation).
\begin{definition}[Maximin Share Allocation]
An allocation $(A_1,\ldots, A_n)$ is said to be a maximin share allocation ($\MMS$) iff for all agents $i \in [n]$ we have $v_i(A_i) \geq \MMS_i$.
\end{definition}

A different solution concept defined by \cite{caragiannis2016unreasonable} requires the maximin share guarantee to hold only for every pair of agents, i.e., an allocation $\mathcal{A} = (A_1, \ldots, A_n)$ is said to achieve \textit{pairwise maximin share guarantee} ($\PMMS$) iff for all  $i, j \in [n]$, we have $v_i(A_i) \geq \mu_i^2(A_i\cup A_j)$ and $v_j(A_j) \geq \mu_j^2(A_i\cup A_j)$.

In this paper, we strengthen $\MMS$ and $\PMMS$, and define a stronger threshold for each agent $i\in [n]$, namely \textit{groupwise maximin share} ($\GMMS_i$). Formally, 

\[\GMMS_i := \underset{J\subseteq[n], J\ni i}{max}\ \ \ \mu_i^{|J|}\left(\bigcup_{j\in J} A_j\right).\]
Now we define \textit{groupwise maximin share allocation} ($\GMMS$ allocation). 
\begin{definition}[Groupwise Maximin Share Allocation]\label{def:GMMS}
An allocation $(A_1,\ldots, A_n)$ is said to be a groupwise maximin share allocation ($\GMMS$) iff for all agents $ i \in [n]$ we have   $v_i(A_i) \geq \GMMS_i$.
\end{definition}

Note that an allocation $(A_1,\ldots, A_n)$ is $\GMMS$ iff $v_i(A_i)~\geq~\mu_i^{|J|}\left(\bigcup_{j\in J} A_j\right)$ for all $J\subseteq[n]$ such that $i\in J$. Also, the threshold $\GMMS_i$ is a function of the underlying allocation and, in contrast, $\MMS_i$ depends only the valuation of agent $i$ for the $m$ goods and the number of agents $n$. 

The fact that there are fair division instances which do not admit an $\MMS$ allocation directly implies that $\GMMS$ allocation are not guaranteed to exist either. Therefore, we consider approximate $\GMMS$ allocations.
\begin{definition}[Approximate Groupwise Maximin Share Allocation]\label{def:approx-GMMS}
An allocation $(A_1,\ldots, A_n)$ is said to be an $\alpha$-approximate groupwise maximin share allocation ($\alpha$-$\GMMS$) iff for all agents $ i \in [n]$ we have $v_i(A_i) \geq \alpha \GMMS_i$.
\end{definition}

A $1$-approximate groupwise maximin share allocation is a $\GMMS$ allocation. 

\subsection{$\GMMS$ Strictly Generalizes $\MMS$ and $\PMMS$}\label{example:GMMS_PMMS_MMS}
This section shows that $\GMMS$ is a distinct solution concept which strictly generalizes both $\MMS$ and $\PMMS$. In fact, the instance constructed in this section shows that there exists allocations which are $\MMS$ and, furthermore, satisfy the maximin share guarantee for all subgroups of size at most $k$ (say), but do not satisfy the $\GMMS$ criteria. 

We now formally define the notion of maximin  share guarantee for a subgroup of size $k$. For ease of presentation, we call such allocations $k$-wise fair. 
An allocation $\mathcal{A} = (A_1, \ldots, A_n)$ is said to be \emph{$k$-wise fair} iff for all agents $i\in [n]$ and for all size-$k$ subsets  $J \subseteq [n]$ such that $i \in J$, the following holds: $v_i(A_i) \geq \mu_i^{k}\left(\bigcup_{j\in J} A_j\right)$. 

In the following example we identify an allocation which is $t$-wise fair---for each $t < k$---but not $k$-wise fair; here, $k \geq 4$.
Let us consider $3k-4$ goods and $n$ agents, with $n>3k-4$. Since the number of agents is greater than the number of goods, $\MMS_i = 0$ for all $i\in [n]$ and, hence, all the allocations are $\MMS$. Consider an agent, say $i=1$, and let her valuation for the $3k-4$ goods be 
\begin{itemize}
\item $3k-7$ for the first $k-1$ goods $l_1, l_2, \ldots, l_{k-1}$ (``large'' goods),
\item $3$ for the next $k-2$ goods $d_2, \ldots, d_{k-1}$ (``medium-sized'' goods), and
\item $1$ for the remaining $k-1$ goods $s_1, ..., s_{k-1}$ (``small'' goods).
\end{itemize}

The valuation of other agents can be set to ensure fairness for them. Let allocation $\mathcal{A} = (A_1,\ldots,A_n)$ be such that $A_1 =\{l_1\}$, $A_2=\{l_2, d_2\}$, $\ldots$, $A_{k-1} = \{l_{k-1}, d_{k-1}\}$, $A_k=\{s_1,...,s_{k-1}\}$ and $A_r = \emptyset$ for all $r \geq k$, i.e., no goods are allocated to the last $n-k$ agents. We first show that the bundle allocated to agent $1$ is at least her pairwise maximin share with any other agent $j\in[n]$. Agent $1$ does not envy agent $k$ since $v_1(l_1)= 3k - 7 > k-1 = v_1(A_k)$ (since $k\geq 4$). Moreover, for the other agents $j\in\{2,\ldots,k-1\}$, we have $\mu_1^{2}(\{l_1, l_j, d_j\})= 3k - 7 = v_1(\{l_1\}) = v_1(A_1)$. Thus, $\mathcal{A}$ ensures $\PMMS$. Below, we present a case analysis which shows that $\mathcal{A}$ is also $t$-wise fair, for all $t\leq k-1$. 
\begin{enumerate}
\item Consider any subgroup $J$ of the first $k-1$ agents, containing agent $1$. Let $|J|=t$. In this case, exactly $2t-1$ goods are allocated to the group of agents $J$. In particular, this set of goods consists of $t$ ``large'' goods and $t-1$ ``medium-sized'' goods. 
Hence, any $t$-partition of these $2t-1$ goods will have at least one bundle with at most one good. Therefore, any partition which maximizes the least valued (with respect to the valuation of agent $1$) bundle will have one bundle containing a ``large'' good, and each of the remaining $t-1$ bundles containing one ``large'' and one ``medium-sized'' good.
Therefore, the allocation $A_1=\{l_1\}$ ensures maximin share guarantee for agent $1$ for all subgroups $J \subseteq \{1,\ldots,k-1\}$, which contain agent $1$. 

\item Next, consider any subgroup $J$ of size $t$ which includes agent $k$ as well as agent $1$. Here, the set of goods allocated to agents in $J$ consists of $t-1$ ``large'' goods, $t-2$ ``medium-sized'' goods and $k-1$ ``small'' goods. To simplify the analysis, we assume that $t=k/2$ and use an averaging argument to obtain 
\begin{align*}
\mu_1^t\left(\cup_{j\in J} A_j\right) &   \leq \frac{(3k-7)(t-1) + 3(t-2) + 1(k-1)}{t} \\
			 &= 3k - 4 - \frac{2k}{t}\\
			 &\leq 3k - 4 - 4\quad \qquad (\textrm{since } t\leq {k}/{2})\\
			 &= 3k - 8\\
			 &< 3k - 7 = v_1(A_1).
\end{align*}

Hence, for all such subgroups of agents of size $t\leq k/2$, the bundle $A_1=\{l_1\}$ satisfies the maximin share guarantee of agent $1$. A similar, but slightly involved argument establishes the result for subgroups of size $t\leq k-1$.

\item Finally, consider a size-$t$ subgroup $J \in [n]$ such that $1 \in J$ and $J$ contains at least one agent $h > k$, i.e., $J$ includes an agent $h$ with an empty bundle $A_h = \emptyset$. Write $J':=J\setminus\{h: h>k\}$ and $|J'|=t'$. The arguments used in the first two cases show that  the maximin guarantee for agent $1$ holds for $J'$, i.e., $v_1(A_1)\geq \mu_1^{t'}(\cup_{j\in J'}\ A_j)$. Using the fact that $t'<t$ and $\cup_{j\in J'}\ A_j = \cup_{j\in J}\ A_j$, we get $\mu_1^{t'}(\cup_{j\in J'}\ A_j)\geq\mu_1^t(\cup_{j\in J}\ A_j)$. These observations establish the required maximin share guarantee, $v_1(A_1)\geq \mu_1^t(\cup_{j\in J}\ A_j)$, for this case as well. 
\end{enumerate}
Overall, we get that the allocation $\mathcal{A}$ satisfies $t$-wise fairness, for any $t<k$. We will complete the analysis by showing that $\mathcal{A}$ is \emph{not} $k$-wise fair. 

In particular, let $J = [k]$. Note that, here, the set $\cup_{j\in J} A_j$ consists of $k-1$ ``large'' goods, $k-2$ ``medium-sized'' goods and $k-1$ ``small'' goods. In addition, the maximin share value with respect to this subgroup is $3k-6$:  partition the goods to obtain one bundle containing all $k-2$ ``medium-sized'' goods (each of value $3$) and, additionally, $k-1$ bundles each containing one ``large'' good and one ``small'' good. Therefore, $\mu_1^k\left(\cup_{j\in J} A_j\right) = 3k - 6 > 3k - 7 = v_1(A_1)$.



This example illustrates that, in particular, $\GMMS$ is a stronger solution concept than $\PMMS$ and $\MMS$. 

\section{$\GMMS$ can be arbitrarily better than $\MMS$}\label{sec:example}
In this section, we provide a class of examples where an $\MMS$ allocation is not necessarily satisfactory in terms of agents' valuations. In particular, we show that imposing $\GMMS$ leads to allocations which Pareto dominate an allocation which only satisfies $\MMS$.

Consider $n$ agents and a set of $n+3$ goods (where $[m]=\{g_1,\ldots,g_{n+3}\}$), along with parameter $V$ and a sufficiently small $\varepsilon$; $0< \varepsilon \ll V $. The valuations are assumed to be additive and are as follows:\\

$\mbox{For agents }\ i \in\{1,\ldots, n-1\}$,\\
$\begin{array}{ll}
& \ \ \ v_i(\{g_z\}) =   \left\{
\begin{array}{ll}
      V, &\quad z \in \{1,\ldots,n-3\} \\
      V/2, & \quad z \in \{n-2, n-1\} \\
      V/2-\varepsilon, & \quad z \in \{n, n+2\} \\
      V/2+\varepsilon, & \quad z \in \{n+1, n+3\} \\
\end{array} 
\right. \\\\
&\mbox{For the agent } n,\\
&\ \ \  v_n(\{g_z\}) =   \left\{
\begin{array}{ll}
      V, &\quad z \in \{1,\ldots,n-1\} \\
      0, & \quad z \in \{n, n+1\} \\
      \varepsilon, & \quad z \in \{n+2, n+3\} \\
\end{array} 
\right.\\
\end{array}$\\

Note that in this fair division instance the maximin share of the first $(n-1)$ agents,  $\MMS_i = V$ for all $i \in \{1,\ldots,n-1\}$. In addition, the maximin share of the last agent, $\MMS_n$, is $2\varepsilon$. 

Now, consider the allocation $\mathcal{A}=(A_1, \ldots,A_n)$ wherein 
$A_i = \{g_{i}\}, \ i \in \{1,\ldots,n-3\}$, 
$A_{n-2}=\{g_{n-2}, g_{n-1}\}$,
$A_{n-1} = \{g_{n}, g_{n+1}\}$, and 
$A_{n} = \{g_{n+2}, g_{n+3}\}$.
Allocation $\mathcal{A}$ is $\MMS$ since $v_i(A_i) = V$ for each $ i \in \{1,\ldots,n-1\}$ and $v_n(A_n) = 2\varepsilon$. In this allocation the valuation of agent $n$ is unsatisfactorily low. Another relevant observation is that this allocation is not $\GMMS$. Below, we show that in this instance \emph{any} $\GMMS$ allocation allocates a bundle of value $V$ to every agent, including agent $n$.  

The fact that $\mathcal{A}$ is not $\GMMS$ follows by considering the goods allocated to agents $n-2$ and $n$, i.e., let $S := A_{n-2}\cup A_n = \{g_{n-2}, g_{n-1}, g_{n+1}, g_{n+3}\}$. Now, $\mu_n^2(S)=V+\varepsilon$, but $v_n(A_n) = 2\varepsilon \ll V$. Furthermore, it can be observed that it is necessary to allocate the agent $n$ at least one of her high valued goods $\{1,\ldots,n-1\}$ to satisfy $\GMMS$. Thus, $\GMMS$ allocation would always ensure that a bundle of at least $V$ is allocated to all the agents. That is, in this instance, unsatisfactorily low valuations can be avoided by imposing the groupwise maximin share guarantee.

\section{Scale of Fairness}\label{sec:fairness-scale}
As mentioned earlier, \textit{envy freeness} ($\EF$) is a well-studied solution concept in the context of fair division of divisible items. However, for indivisible goods, a simple example with one positively valued good and two agents shows that envy-free allocations do not always exist. Hence, for indivisible goods, natural relaxations of envy freeness---in particular, $\EFone$ and $\EFX$---have been considered in the literature. We now provide formal definitions of these relaxations. 

\begin{definition}[Envy-free up to one good~\cite{budish2011combinatorial}]\label{def:EF1}
An allocation $\mathcal{A}=(A_1,A_2,\dots,A_n)$ is said to be {envy-free up to one good} ($\EFone$) iff for every pair of agents $i , j \in [n]$ there exists a good $g \in A_j$ such that $v_i(A_i) \geq v_i(A_j\setminus\{g\})$.
\end{definition}

\begin{definition}[Envy-free up to the least positively valued good \cite{caragiannis2016unreasonable}]\label{def:EFX}
An allocation $\mathcal{A}=(A_1,A_2,\dots,A_n)$ is said to be {envy-free up to the least valued good} ($\EFX$) iff for every pair of agents $i ,j \in [n]$ and for all goods $g \in A_j \cap \{ g' \in [m] \mid v_i(g') >0 \}$ (i.e., for all goods $g$ in $A_j$ which are positively valued by agent $i$) we have $v_i(A_i)\geq v_i(A_j\setminus\{g\})$.
\end{definition}

Note that the above mentioned definitions imply that, for additive valuations, an $\EFX$ allocation is necessarily $\EFone$ as well. Next we show that, interestingly, these relaxed versions of envy freeness are implied by $\GMMS$, but not by $\MMS$.

\begin{proposition}[Scale of Fairness]\label{lem:scale-envy}
In any fair division instance with additive valuations
	\begin{enumerate}
		\item If an allocation $\mathcal{A}$ is envy free ($\EF$) then it achieves groupwise maximin share guarantee ($\GMMS$) as well. 
		\item If an allocation $\mathcal{B}$ is $\GMMS$ then it is $\EFX$ (and, hence, $\EFone$) as well. 
	\end{enumerate}
\end{proposition}
\begin{proof}
First we will show that $\EF$ implies $\GMMS$: Assume that the allocation $\mathcal{A} = (A_1, \ldots, A_n)$ is {envy free}, that is, for all $i ,j \in [n]$ we have   $v_i(A_i) \geq v_i(A_j)$. Therefore, for any agent $i$ and any group of agents $J \subseteq [n]$ we have $|J| v_i(A_i) \geq \sum_{ j \in J} v_i(A_j)$. Since the valuation $v_i$ is additive, this inequality leads to the following bound $v_i(A_i) \geq \frac{1}{|J|} v_i(S)$; here $S = \cup_{j \in J} A_j$. Now, an averaging argument establishes the following inequality: for any $|J|$-partition of $S$ if $(P_1, P_2, \ldots, P_{|J|}) \in \Pi_{|J|} (S)$ then $v_i(A_i) \geq \min_{1 \leq k \leq |J|} v_i(P_k)$.
Hence, $v_i(A_i) \geq \max_{(P_1, \ldots, P_{|J|}) \in \Pi_{|J|} (S)} \ \min_{k} v_i(P_k) = \mu_i^{|J|} (S)$, and we get that $\A$ is $\GMMS$. 

Next, we argue that $\GMMS$ implies $\EFX$:  \cite{caragiannis2016unreasonable} have shown that  that $\PMMS$ implies $\EFX$. By definition, $\GMMS$ implies $\PMMS$. Hence, a $\GMMS$ allocation is guaranteed to be $\EFX$.
\end{proof}
Note that \cite{caragiannis2016unreasonable} also provided an example to show that maximin share guarantee by itself does not imply $\EFone$. Hence, an $\MMS$ allocation is not necessarily $\EFX$. Consider a fair division instance with three agents, five goods, and each agent values each good at $1$. The maximin share of all the agents is $1$. Thus, the allocation $A_1=\{g_1,g_2,g_3\}, A_2=\{g_4\}$ and $A_3=\{g_5\}$ satisfies $\MMS$, but not $\EFone$. This is because agents $2$ and $3$ continue to envy agent $1$ even if a single good is removed from $A_1$.

This, overall, shows that while $\MMS$ is not enough to guarantee weaker notions of envy freeness, $\GMMS$ ensures fairness in terms of such notions, and secures a place in the scale of fairness. \\

Next we will consider the complementary direction of going from bounded envy to groupwise maximin fairness. In particular, we will establish existence and algorithmic results for approximate $\GMMS$ by considering a solution concept which is stronger than $\EFone$, but weaker than $\EFX$. Specifically, we will define allocations which are \emph{envy-free up to one less-preferred good} ($\EFL$)---see Definition~\ref{def:EFL} below---and show that such allocations are guaranteed to exist, when the valuations are additive. Note that, in contrast, the generic existence of $\EFX$ allocations remains an interesting open question. Furthermore, we will prove that $\EFL$ allocations can be computed in polynomial time and, under additive valuations, such allocations imply $1/2$-$\GMMS$. 

\begin{definition}\label{def:EFL}
An allocation $\mathcal{A}=(A_1,A_2,\dots,A_n)$ is said to be \emph{envy-free up to one less-preferred good} ($\EFL$) if for every pair of agents $i , j  \in [n]$ at least one of the following conditions hold:
\begin{itemize}
\item $A_j$ contains at most one good which is positively valued by $i$; $|A_j \cap \{ g' \mid v_i(g') > 0 \} | \leq 1 $  
\item There exists a good $g \in A_j$ such that $v_i(A_i) \geq v_i(A_j\setminus\{g\})$ and $v_i(A_i) \geq v_i(\{g\})$.
\end{itemize}
\end{definition}

The fact that an $\EFL$ allocation is $\EFone$ follows directly from the definitions of these solution concepts. Also, note that if an allocation $(A_1, \ldots, A_n)$ is $\EFX$ then for any pair of agents $i, j \in [n]$ with $| A_j \cap \{ g' \mid v_i(g') > 0 \} | > 1 $ the second condition in the definition of $\EFL$ holds. In particular, write $A_j^i := A_j \cap  \{ g' \mid v_i(g') > 0 \} $ and consider two distinct goods $\hat{g}, \tilde{g} \in A_j^i$. Since the allocation is $\EFX$, we have $v_i(A_i) \geq v_i(A_j \setminus \{\hat{g}\})$ and $v_i(A_i) \geq v_i(A_j \setminus \{\tilde{g}\})$. Note that $\hat{g} \in A_j \setminus \{\tilde{g}\}$, and, hence, $v_i(A_i) \geq v_i(\hat{g})$. This implies that good $\hat{g}$ satisfies the second condition in Definition~\ref{def:EFL}. Hence, any $\EFX$ allocation is  $\EFL$ as well.  

With this new fairness notion, we have the following chain of implications:  $\mathrm{EF} \Rightarrow \GMMS \Rightarrow \EFX \Rightarrow \EFL \Rightarrow \EFone$.

\section{An Approximation Algorithm for $\bGMMS$}
Our main result in this section shows that $1/2$-$\GMMS$ allocations always exist under additive valuation, and such allocations can be found efficiently.  

\begin{theorem}\label{thm:main}
Every fair division instance with additive valuations admits a $1/2$-approximate groupwise maximin share allocation. Furthermore, such an allocation can be found in polynomial time. 
\end{theorem}

\noindent \textit{Proof-Sketch}
The proof proceeds in two steps. First, we provide a constructive proof for the existence of  $\EFL$ allocations, under additive valuations (Section \ref{sec:EFL-existence}). Next, we complete the proof by showing that $\EFL$ implies $1/2$-$\GMMS$ (in Section \ref{sec:EFL-GMMS}) \hfill \qed

\subsection{Existence of $\bEFL$ Allocations} \label{sec:EFL-existence}
This section shows that $\EFL$ allocations are guaranteed to exist when the valuations are additive. Specifically, we develop an algorithm that always finds such an allocation.

\begin{lemma} \label{lemma:EFL}
Given any fair division instance with additive valuations, Algorithm~\ref{alg:envygraph} finds an $\EFL$ allocation in polynomial time. 
\end{lemma}

Algorithm~\ref{alg:envygraph} iteratively allocates the goods and maintains a partial allocation, $\mathcal{A}=(A_1, \ldots, A_n)$, of the goods assigned so far. In each iteration, the algorithm selects an agent $i$ who is not currently envied by any other agent, and allocates $i$ an unassigned good of highest value (under $v_i$). 

{
	\begin{algorithm}
	{
		{{\bf Input :} $n$ agents, $m$ items, and valuations $v_i{\{g\}}$ for each agent $i \in [n]$ and for each good $g \in [m]$.\\ {\bf Output:} An $\EFL$ allocation.}
		\caption{Finding an $\EFL$ Allocation}
		\label{alg:envygraph}
		\begin{algorithmic}[1]
			\STATE Initialize allocation $\mathcal{A}=(A_1, A_2, \ldots, A_n)$ with $A_i = \emptyset$ for each agent $ i \in [n]$, and $M=[m]$.
			\STATE Initialize directed \textit{envy} graph $G(\mathcal{A})=([n], E)$ where $E=\emptyset$. 
			\WHILE{$M \neq \emptyset$}
			\STATE Pick a source vertex $i$ of ${G}(\mathcal{A})$. \COMMENT{such a vertex always exists, since ${G}(\mathcal{A})$ is acyclic.}\\
			\STATE Pick $g\in \argmax_{j\in M} v_i{\{g\}}$.
			\STATE Update $A_i \leftarrow A_i \cup \{ g \}$ and $M \leftarrow M \setminus \{g \}$.
			\WHILE{the current envy graph $G(\mathcal{A})$ contains a cycle} 
			\STATE Update $\mathcal{A}$ (using {Lemma~\ref{lemma:envylipton}}) to remove the cycle.
			\ENDWHILE
			\ENDWHILE 
			\STATE Return $\mathcal{A}$.
		\end{algorithmic}}
	\end{algorithm}}
	
	Throughout the execution of the algorithm, the existence of an unenvied agent is ensured by maintaining a directed graph, $G(\mathcal{A})$, that captures the envy between agents. The nodes in this envy graph represent the agents and it contains a directed edge from $i$ to $j$ iff $i$ envies $j$, i.e., $v_i(A_i) < v_i(A_j)$. Lemma \ref{lemma:envylipton}, established in~\cite{lipton-envy-graph}, shows that if any iteration leads to a cycle in the envy graph $\mathcal{G}(\mathcal{A})$, then we can always \emph{resolve} it to obtain an acyclic envy graph without decreasing the valuation of any agents; for completeness, we provide a proof of Lemma \ref{lemma:envylipton}. It is relevant to that, since $G(\mathcal{A})$ is acyclic for a partial allocation $\mathcal{A}$, it necessarily contains a source node, i.e., an agent $i$ who is not envied by other agents. 

Although, Algorithm $\ref{alg:envygraph}$ is similar to the algorithm developed in \cite{lipton-envy-graph}---which efficiently finds $\EFone$ allocations--- here, instead of assigning goods in an arbitrary order, we always allocate to an unenvied agent the available good she values the most. This is crucial for obtaining an $\EFL$ allocation.
\begin{lemma}\cite{lipton-envy-graph}
	\label{lemma:envylipton}
	Given a partial allocation $\mathcal{A}=(A_1, \ldots, A_n)$ of a subset of goods $S \subseteq [m]$, we can find another partial allocation $\mathcal{B}=(B_1, \ldots, B_n)$ of $S$ in polynomial time such that \\
(i) The valuations of agents for their bundles do not decrease, that is, $v_i(B_i) \geq v_i(A_i)$ for all $i \in [n]$.\\
(ii) The envy graph $G(\mathcal{B})$ is acyclic.
\end{lemma}

\begin{proof}
	If the envy graph of $\mathcal{A}$ is acyclic then the claim holds trivially. Otherwise, let us denote a cycle in the graph $G(\mathcal{A})$ by $C=i_1 \rightarrow i_2 \rightarrow \ldots \rightarrow i_k \rightarrow i_1$. Now, we can reallocate the bundles as follows: for all agents not in $C$, i.e., $k \notin \{i_1, i_2, \ldots, i_k \}$ set $A'_k = A_k$, and for all the agents in the cycle set $A'_i$ to be the bundle of their successor in $C$, i.e., set $A'_{i_a} = A_{i_{(a+1)}}$ for $ 1\leq a < k$ along with $A'_{i_k} = A_{i_1}$.
	
	Note that after this reallocation we have $v_i(A'_i) \geq v_i(A_i)$ for all $i \in [n]$. Furthermore, the number of edges in $G(\mathcal{A})$ strictly decreases: The edges in $C$ do not appear in the envy graph of $(A'_1, \ldots, A'_n)$ and if an agent $k$ starts envying an agent in the cycle, say agent $i_a$, then $k$ must have been envious of $i_{a+1}$ in $\mathcal{A}$. Edges between agents $k$ and $k'$ which are not in $C$ remain unchanged, and edges going out of an agent $i$ in the cycle $C$ can only get removed, since $i$ valuation for the bundle assigned to her bundle increases. Therefore, we can repeatedly remove cycles and keep reducing the number of edges in the envy graph to eventually a find a partial allocation $\B$ that satisfies the stated claim.
\end{proof} 

\begin{proof}[Proof of Lemma \ref{lemma:EFL}]
Write $\mathcal{A} = (A_1, \ldots, A_n)$ to denote the allocation returned by Algorithm~\ref{alg:envygraph}.
First, we note that an inductive argument proves that $\mathcal{A}$ is $\EFone$. In fact, we will show that the $\EFone$ condition holds for $\mathcal{A}$ with respect to the last (in terms of the algorithm's allocation order) good $g$ assigned to each bundle $A_j$. Write $g_t$ to denote the good allocated in the $t$th iteration of Algorithm~\ref{alg:envygraph}; hence, the goods are allocated in the following order $g_1, g_2, \ldots, g_m$.   

The initial partial allocation $(\emptyset,\emptyset,\dots,\emptyset)$ is $\EFone$ (in fact, it is envy free). Now, say that in the $j$th iteration the algorithm allocates good $g_j$ to agent $i$. Write $\mathcal{A}^{j} =(A^{j}_1, \ldots, A^{j}_n)$ and $\mathcal{A}^{j+1} = (A^{j+1}_1, \ldots, A^{j+1}_n)$, respectively, to denote the partial allocations before and after the $j$th iteration. The induction hypothesis implies that $\mathcal{A}^{j}$ is $\EFone$ with respect to the last assigned good. Therefore, for every pair of agents $r, s \in [n]$, we have $v_r(A^{j}_r)\geq v_r(A^{j}_s\setminus\{g_a\})$, where $g_a$ is the last good assigned to the bundle $A_s^j$ (i.e., for any other good $g_b \in A^j_s$, we have $b<a$).   

Since the good $g_j$ is allocated to agent $i$, it must be the case that $i$ is a source vertex in $G(\mathcal{A}^j)$, i.e., no agent envies $i$ under $\mathcal{A}^j$. 
This implies that, $v_r(A^j_r) \geq v_r(A^j_i) = v_r((A^j_i\cup\{g_{j}\})\setminus\{g_{j}\})$ for all $r \in [n]$. Note that at this point of time, $g_j$ is the last good assigned to the bundle $A^j_i$. In addition, from the proof of Lemma~\ref{lemma:envylipton}, we know that $\mathcal{A}^{j+1}$ is a permutation of the allocation $(A^j_1,A^j_2,\ldots,A^j_i\cup\{g_{j}\},\ldots,A^j_n)$, and $v_r(A^{j+1}_r) \geq v_r(A^j_r)$ for all $r \in [n]$. Hence, for every pair of agents $r, s \in [n]$ there exists a good $g_a \in A^{j+1}_s$ such that $v_r(A^{j+1}_r)\geq v_r(A^{j+1}_s\setminus\{g_a\})$. In addition, $g_a$ is the last good assigned to the bundle $A^{j+1}_s$. That is, the stated property holds for $\mathcal{A}^{j+1}$ as well. 

Now, we will use this observation to prove that $\mathcal{A}$ is $\EFL$. Specifically, we show the $\EFL$ conditions hold for, say, agent $1$ (analogous arguments establish the claim for the other agents). Suppose that, during the execution of the algorithm, agent $1$ receives its first good, $g_t$, in the $t$th iteration. Note that the partial allocation before the $t$th iteration, say $\mathcal{A}^t=(A^t_1, \ldots, A_n^t)$, satisfies $|A^{t}_i \cap \{g' \in [m] : v_1(g')>0\}| \leq 1$ for all $i \in [n]$. This bound follows from the observation that during any previous iteration $s < t$, the selected source vertex $i$ (i.e., the agent that gets a new good during the $s$th iteration) does not contain any good which is positively valued by agent $1$; otherwise, $i$ would have been envied by $1$, contradicting the fact that it is a source vertex. Hence, each bundle in $\mathcal{A}^t$ contains at most one good from the set $\{g' \in [m] : v_1(g')>0\}$.

Let us now consider the final allocation $\mathcal{A}=(A_1, \ldots, A_n)$ and any agent $j \in [n]$. If $|A_j \cap \{g' \in [m] : v_1(g')>0\}| \leq 1$, then the first condition in the definition of $\EFL$ holds and we are done, else if $|A_j \cap \left\{g' \in [m] : v_1(g')>0\right\}| > 1$, then bundle $A_j$ must have received a good after the $t$th iteration. This is consequence of the above-mentioned property of the partial allocation $\mathcal{A}^t$. Write $g_\ell$ to denote the last good allocated to the bundle $A_j$. We have $\ell > t$, since $g_\ell$ was assigned to $A_j$ after the $t$th iteration. Also, in the $t$th iteration good $g_t$ was selected instead of $g_\ell$, hence it must be the case that $v_1(g_t) \geq v_1(g_\ell)$.  

Note that, as mentioned before, the $\EFone$ condition holds for $A_j$ with respect to $g_\ell$, i.e., $v_1(A_1) \geq v_1(A_j \setminus \{ g_\ell\})$. In addition, Lemma~\ref{lemma:envylipton} implies that $v_1(A_1) \geq v_1(\{g_t\})$. Therefore, if $|A_j \cap \{g' \in [m] : v_1(g')\}| >1$ for any $j \in [n]$, then there exists a good $g_\ell$ such $v_1(A_1)\geq v_1(A_j\setminus\{g_\ell\})$ and $v_1(A_1) \geq v_1(g_\ell)$. Hence, $\mathcal{A}$ is an $\EFL$ allocation.
\end{proof}

\subsection{$\bEFL$ implies Approximate $\bGMMS$}\label{sec:EFL-GMMS}
\begin{lemma}\label{lem:EFL-GMMS}
In any fair division instance with additive valuations, if an allocation $\mathcal{A}=(A_1, \ldots, A_n)$ is $\EFL$, then it is $1/2$-$\GMMS$ allocation as well. 
\end{lemma}

\begin{proof} 
Fix an agent $i$ and a set of $k$ agents, $J \subseteq [n]$ which contains $i$, i.e., $|J|=k$ and $ J \ni i$. Also, let $S = \cup_{j \in J} A_j$ denote the set of all the goods allocated to the agents in $J$. We will show that agent if $(A_1, \ldots, A_n)$ is $\EFL$, then $v_i(A_i)$ is at least $\frac{1}{2}$ times the maximin share of $i$ restricted to $S$, i.e., $v_i(A_i) \geq \frac{1}{2} \mu_i^k(S)$. This establishes the stated claim. 

Write $S_i \subseteq [m]$ to denote the set of goods which are positively valued by $i$,  $S_i := \{ g' \in [m] \mid v_i(g') >0 \}$. Now, among the set of agents $J \setminus \{i \}$ consider the ones who are allocated at most one good from $S_i$; specifically, let $T := \{j \in J \setminus \{ i \} \mid  |A_j \cap S_i | \leq 1\}$.
Write $J' := J \setminus T$, $t' := |J'|$, and $S' = \cup_{j \in J'} A_j$. Note that the agent $i$ belongs to the group $J'$, and for all $j \in J'$ we have $|A_j \cap S_i | >1$. Therefore, the fact that $\mathcal{A}$ is $\EFL$ implies that, for all $j \in J'$, there exists a good $g_{(j)} \in A_j$ such that $v_i(A_i) \geq v_i(A_j\setminus \{g_{(j)}\})$ and $v_i(A_i)\geq v_i(g_{(j)})$. In other words, for the additive valuation $v_i$, we have $2 v_i(A_i) \geq v_i(A_j)$.

We will now establish the multiplicative bound with respect to $\mu_i^{t'}(S')$ (the maximin share of $i$ restricted to $S'$) and prove that $\mu_i^{t'}(S') \geq \mu_i^t (S)$. This will complete the proof. Since $v_i$ is additive, an averaging argument gives us 

\begin{align*}
\mu_i^{t'}(S') &  \leq \frac{1}{t'} v_i(S') = \frac{1}{t'} \sum_{j \in J'} v_i(A_j) \leq \frac{1}{t'} 2 t' v_i(A_i)
\end{align*}
Here, the last inequality uses the bound $2 v_i(A_i) \geq v_i(A_j)$ for all $j \in J'$. Therefore, $v_i(A_i) \geq \frac{1}{2} \mu_i^{t'}(S')$. To complete the proof, we need to show that $\mu_i^{t'}(S') \geq \mu^t_i (S)$. Note that---while considering the maximin shares of agent $i$---we can restrict our attention to $S_i$ (the set of goods which are positively valued by $i$). In particular, the equalities $\mu_i^{t'}(S') = \mu_i^{t'}(S' \cap S_i)$ and $\mu_i^{t}(S) = \mu^t_i(S \cap S_i)$ imply that, without loss of generality, we can work under the assumptions that $S \subseteq S_i$ and $S' \subseteq S_i$. Effectively, for all $j \in J$, we can slightly abuse notation and denote $A_j \cap S_i$ by $A_j$. Now, consider the allocation $ \mathcal{M} = (M_1,\ldots, M_t) \in \underset{(B_1, \ldots, B_t) \in \Pi_t(S)}{\argmax} \displaystyle \ \min_{j \in [t]} v_i (B_j)$. 

We have, $ \min_{j \in J} v_i(M_j)= \mu_i^t(S)$. Also, note that $|A_j | \leq 1$ for each agent $j \in K = J \setminus J'$. Therefore, there are at most $|K| = t-t'$ bundles in $\mathcal{M}$ with items from $S \setminus S'$. We choose $t'$ bundles from $ \mathcal{M} = (M_1,\ldots, M_t)$ which do not contain any item from $S \setminus S'$. Let us call these new bundles $\mathcal{M}'=(M'_1, \ldots, M'_{t'})$. By the definition of $\mathcal{M}$, the value of each such bundle for agent $i$ is greater than or equal to $\mu_i^t(S)$, that is, $v_i(M'_{j})\geq \mu_i^t(S)$ for all $  j \in [t']$. Since we have assumed that all items have nonnegative values, adding more items from the remaining $(t-t')$ bundles to any of the bundles in $\mathcal{M}'$ can only increase the value of the partitions. Thus, $\mu_i^{t'}(S') \geq \mu_i^t(S)$, which implies that $v_i(A_i) \geq \frac{1}{2} \mu_i^t(S)$.		
\end{proof}

Lemmas~\ref{lemma:EFL} and~\ref{lem:EFL-GMMS} directly establish Theorem~\ref{thm:main}. 

We note that (both in terms of the algorithm and $\EFL$'s implication) the approximation guarantee established in Theorem~\ref{thm:main} is essentially tight. Appendix~\ref{sec:tightExample} provides an example to establish this observation.

\section{Guaranteed Existence of $\bGMMS$ Allocations in Restricted Settings}
In this section, we prove the existence of $\GMMS$ for relevant valuation classes. 
\begin{theorem}
Groupwise maximin share allocations always exist under additive, binary valuations, i.e., such allocations exist when the additive valuations satisfy $v_i(g) \in\{0,1\}$ for all agents $ i \in[n]$ and goods $ g\in[m]$. 
\end{theorem}

\begin{proof}
To prove the theorem, it suffices to show that under additive, binary valuations, $\EFone$ implies $\GMMS$. Recall that $\EFone$ allocations are guaranteed to exist. Let $\mathcal{A}=(A_1,\dots,A_n)$ be an $\EFone$ allocation. Fix an agent $i$ and a group of agents $J \ni  i$. Also, write $S := \bigcup_{j\in J} A_j$. Next we complete the proof by showing that $v_i(A_i) \geq \mu_i^{|J|}(S)$. 
	
Since the valuations are binary, for any subset of goods $B \subseteq [m]$, the total value $v_i(B)$ is a non-negative integer. Therefore, $\mu_i^{|J|}(S)$ must be a non-negative integer. Moreover, since the valuations are additive, we have $\frac{1}{|J|}{v_i(S)} \geq \mu_i^{|J|}(S)$. In addition, the facts that (a)~$\mathcal{A}$ is an $\EFone$ allocation and (b)~the maximum value of any good is $1$ imply $v_i(A_i)\geq v_i(A_j) - 1$ for all $j \in J$. Therefore,

\begin{align*}
v_i(A_i) &   = \frac{1}{|J|}{|J| v_i(A_i)} \geq \frac{1}{|J|}\left(\!\!v_i(A_i) +\!\!\!\! \sum_{j\in J\setminus\{i\}}\!\!\!\!(v_i(A_j) - 1)\!\! \right)\\
			 &\geq \frac{1}{|J|}(v_i(S)-|J|+1)   = \frac{1}{|J|}v_i(S) - \frac{|J|-1}{|J|}\\
			 &\geq \mu_i^{|J|}(S) - \frac{|J|-1}{|J|} > \mu_i^{|J|}(S) - 1.
\end{align*}
The last inequality $v_i(A_i) - \mu_i^{|J|}(S)>-1$ implies $v_i(A_i) \geq \mu_i^{|J|}(S)$ since $v_i(A_i)$ and $\mu_i^{|J|}(S)$ are integers. Overall, the inequality $v_i(A_i) \geq \mu_i^{|J|}(S)$ shows that the $\EFone$ allocation $\mathcal{A}$ is $\GMMS$ as well. This completes the proof.
\end{proof}
Note that the following theorem holds for general combinatorial valuations and is not limited to the additive case. 
\begin{theorem}
If all the $n$ agents in a fair division instance have the same valuation (i.e., $v_i = v_j$ for all $i,j \in [n]$), then a groupwise fair allocation is guaranteed to exist. 
\end{theorem}
\begin{proof}
Given a fair division instance with $m$ goods and $n$ agents with the same valuation function $v$, we will define an order on $\Pi_n([m])$, the set of $n$-partitions of $[m]$. For vectors $u, u' \in \mathbb{R}^n$, we say that $u$ \emph{lexicographically dominates} $u'$ if $u=u'$, or there exists an index $a \in [n]$ such that $u_{(a)} > u'_{(a)}$ and for all $b < a$ we have $u_{(b)} = u'_{(b)}$; here, $u_{(k)}$ and $u'_{(k)}$ denotes the $k$th smallest component of $u$ and $u'$, respectively.  Extending this definition, an allocation $(A_1, \ldots, A_n) \in \Pi_n([m])$ is said to  lexicographically dominate another allocation $(B_1, \ldots, B_n) \in \Pi_n([m])$ iff the vector of valuations $(v(A_1), \ldots, v(A_n))$ lexicographically dominates $(v(B_1), \ldots, v(B_n))$.  
Note that lexicographic domination defines a total order over the equivalence classes of $n$-partitions. In addition, up to permutations of the valuation vector, there exists a unique maximal allocation with respect to this order, i.e., there exists an allocation $\mathcal{A}^*=(A^*_1, \ldots, A^*_n)$ which lexicographically dominates all other allocations. 

We will show that $\mathcal{A}^*$ is $\GMMS$. For contradiction, say that this is not the case, and the $\GMMS$ condition is violated for subset $J \subseteq [n]$ and agent $i \in J$, i.e., $v_i(A^*_i) < \mu_i^{|J|}\left(S \right)$; here $S:=\cup_{j \in J} A^*_j$. Since the valuations are identical, the maximin share of each agent in $J$ restricted to $S$ is the same, i.e., $\mu_i^{|J|}\left(S \right) = \mu_k^{|J|}\left( S \right) $, for all $k \in J$. Hence, $\GMMS$ condition must also be violated for agent $\argmin_{j \in J} v(A^*_j)$. This observation and the definition of maximin shares imply that there exists a $|J|$-partition of $S$, say $(M_1, \ldots, M_{|J|})$, such that $\min_{1\leq a \leq |J|} v(M_a) > \min_{j \in J} v(A_j^*)$. 

Now, consider an allocation $\mathcal{B}=(B_1, \ldots, B_n)$ obtained from $\mathcal{A}^*$ by (i) replacing the $|J|$ bundles $A_j^*$s (for $j \in J$) with $M_a$s (for $1 \leq a \leq |J|$), and (ii) setting $B_k = A^*_k$ for all $k \notin J$. Note that, even under any permutation of the bundles, the allocation $\mathcal{B} \neq \mathcal{A}^*$. Also, the construction of $\mathcal{B}$ ensures that it lexicographically dominates $\mathcal{A}^*$.  This contradicts the maximality (under the defined lexicographic order) of $\mathcal{A}^*$ and, hence, the stated claim follows. 
\end{proof}

\section{Some Empirical Results}
For an experimental investigation, we generated $1000$ random instances, for several combinations of $n$ agents and $m$ goods ($n$ ranging from $3$ to $5$, and $m$ from $3$ to $11$), by randomly drawing the valuations from two different distributions (a)~gaussian distribution with mean $0.5$, standard deviation $0.1$ (truncated at $0$ to ensure nonnegative valuations) and (b)~uniform distribution $[0,1]$. These are the same set of experiments that \cite{bouveret2014characterizing} carried out for studying the existence of various fairness notions under additive valuations.

In addition, we considered the modification of these instances wherein all the agents agree on the same ranking of the goods, but may have different valuations for the same item. Such instances are said to have same-order preferences ($\mathrm{SOP}$), and were studied by \cite{bouveret2014characterizing}. They showed that, when it comes to (empirically) finding an $\MMS$ allocation, $\mathrm{SOP}$ instances are the hardest. We observe similar results for $\GMMS$: finding a $\GMMS$ allocation, done using brute-force search, requires noticeably more time in $\mathrm{SOP}$ instances, than in non-$\mathrm{SOP}$ instances. Our empirical results and observations are summarized below:
\begin{enumerate}
\item $\GMMS$ allocations exist empirically in all randomly generated instances (which is similar to the experimental results for $\MMS$ \cite{bouveret2014characterizing}). These simulation results indicate that we may not fall short on such generic existence results by strengthening the maximin solution concept. 
\item Allocations generated by the $\EFL$ algorithm (Algorithm \ref{alg:envygraph}) on the random instances achieve an approximation factor of $0.98$ (on an average) which is better than our theoretically obtained worst-case bound of $0.5$. Figure (\ref{fig:alphaGMMS}) tabulates the approximation factor for $\GMMS$ using the $\EFL$ algorithm, averaged over $1000$ instances.
\begin{figure}[!ht]
    \centering
    \includegraphics[]{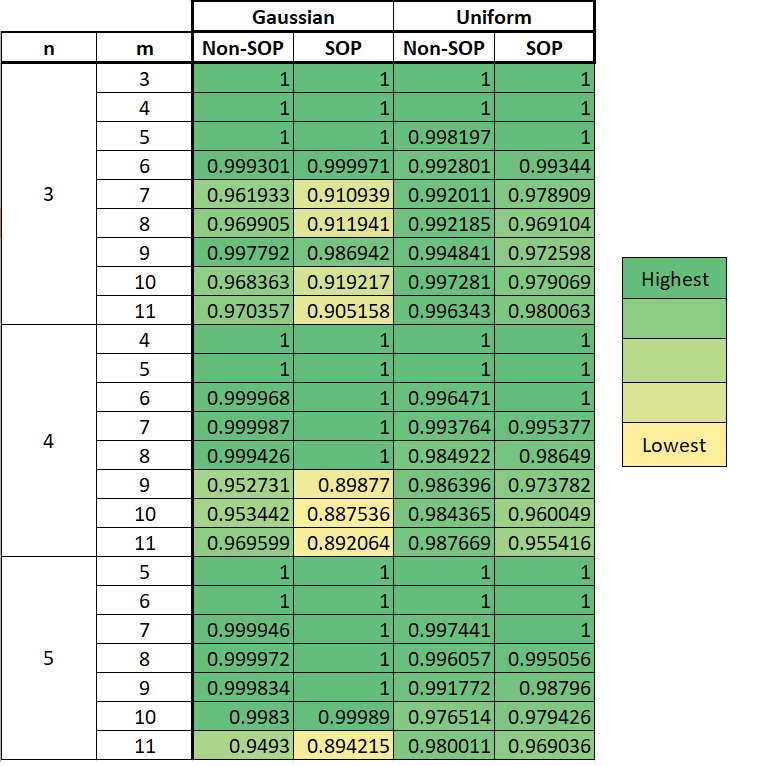}
    \caption{$\EFL$ algorithm results on approximate $\GMMS$. \label{fig:alphaGMMS}}
\end{figure}
The approximation factor (on average) is lower for $\mathrm{SOP}$ than for non-$\mathrm{SOP}$ instances. Moreover, we observe that, for randomly generated instances, the approximation factor is better when the number of items is a multiple of the number of agents. This may be because of the round robin nature of the algorithm which gives an agent her most desirable good at each round.
\item As expected, our proposed $\EFL$ algorithm ran much faster than the (exponential-time) algorithm for finding exact $\GMMS$. Our algorithm was about $10^7$ times faster than the exact $\GMMS$ computation on a machine with a quad Intel Core $i7$ processor and $32$ GB RAM.
\end{enumerate} 

\section{Conclusion and Future Work}
In this paper, we introduced the concept of $\GMMS$ to address fair allocation of indivisible goods, thereby strengthening the well-studied notions of $\MMS$ and $\PMMS$. We established the existence of $1/2$-$\GMMS$ under additive valuations, and developed an efficient algorithm for finding such allocations. We also proved that under specific settings exact $\GMMS$ allocations always exist. In addition, $\GMMS$ allocations were always found in our experiments (over a large number of randomly generated instances). This indicates why it seems nontrivial to obtain simpler\footnote{The intricate examples showing the nonexistence of $\MMS$ under additive valuations \cite{procaccia2014fair,kurokawa2016can} also serve as counterexamples for $\GMMS$.} examples which refute the guaranteed existence of $\GMMS$ allocations. Finding an instance which admits an $\MMS$ allocation but not a $\GMMS$ allocation remains an interesting, open problem. 

Our work addresses key questions about groupwise fairness when the valuations are additive, and it motivates the study of this notion in more general settings, e.g., under submodular valuations, or under additional constraints, such as the ones considered by~\cite{fair-graph}. Another interesting direction of future work is to understand groupwise fair division of indivisible goods among strategic agents.

\subsubsection*{Acknowledgments.} Siddharth Barman was supported by a Ramanujan Fellowship (SERB - {SB/S2/RJN-128/2015}). Arpita Biswas gratefully acknowledges the support of a Google PhD Fellowship Award.

\bibliographystyle{apalike}
\bibliography{MMSbib}

\appendix
\section*{Appendix}
\section{Tight Example for Algorithm~\ref{alg:envygraph}}\label{sec:tightExample}

This section shows that the approximation guarantee established in Theorem~\ref{thm:main} is essentially tight. Specifically, consider the following fair division instance with $n$ agents and ($3n-2$) goods. The additive valuation of agent $1$ is as follows:   
\begin{itemize}
\item For each of the first $n$ goods---say, $\{l_1, l_2, \ldots, l_n\}$---the valuation of agent $1$ is equal to one. 
\item The next $(n-1)$ goods, $\{d_2,d_3, \ldots, d_n\}$, are valued at $(n-1)/n$ by agent $1$.  
\item Finally, the remaining $(n-1)$ goods, $\{s_2,s_3, \ldots, s_n\}$,  are valued at $1/n$. 
\end{itemize}

The remaining $(n-1)$ agents value $l_1$ at $2$, and they value $\{s_2,\ldots,s_n\}$ at $1$. In addition, goods $\{d_2, d_3, \ldots, d_n\}$ along with $\{l_2,\ldots, l_n\}$ have a value of $1/2 - 1/n$ for these $(n-1)$ agents. 

Consider allocation $\mathcal{A}=(A_1,...,A_n)$ wherein the first agent gets bundle $A_1 = \{l_1\}$, and the remaining agents obtain $A_k = \{l_k, d_k, s_k\}$ for all $k>1$. This allocation is $\EFL$ with respect to all the agents. But, the maximin share of agent $1$ is close to $2$; specifically, in the following $n$-partition the value of each bundle for agent $1$ is equal to $2 - 1/n$:  $\{l_1, s_2, \ldots, s_n\}, \{l_2, d_2\}, \ldots,\{l_n, d_n\}$.

Hence, as $n$ increases, the approximation guarantee of this $\EFL$ allocation for agent $1$ approaches $1/2$. 

Furthermore, we can ensure that the proposed algorithm finds this $\EFL$ allocation. Specifically, during the algorithm's execution, say, agent $1$ picks $l_1$ at the beginning. The other agents now pick in the remaining iterations. Since the other agents value the goods $s_j$s the most among the unallocated goods, agents $2$ to agent $n$ will select $s_2$ to $s_n$ during the initial iterations. At this point of time, the partial allocation is $\{l_1\}, \{s_2\}, \ldots,\{s_n\}$.

Agent $1$ is still envied the most. Now, agents $2$ to $n$ pick $d_2$ to $d_n$ one after the other (the ties are broken in favor of the target allocation). The intermediate allocation obtained by the algorithm is $\{l_1\}, \{s_2, d_2\}, \ldots,\{s_n, d_n\}$. Still, Agent $1$ is the most envied one and she does not envy others. Goods $l_2$ to $l_n$ remain unallocated. Agents $2$ to $n$ finally select them, one after the other, to give us the required $\EFL$ allocation.
\end{document}